\documentclass[11pt]{article}

\usepackage[margin=1in]{geometry}
\usepackage{amsmath,amssymb,amsthm}
\usepackage{xcolor}
\usepackage{graphicx}
\usepackage{microtype}
\usepackage[hidelinks]{hyperref}

\usepackage{braket}
\usepackage{mathtools}

\usepackage{array}
\usepackage{booktabs}
\usepackage{tabularx}

\usepackage{enumitem}

\usepackage{tikz}
\usetikzlibrary{arrows.meta, decorations.pathreplacing, positioning, calc, fit, backgrounds}

\usepackage[most]{tcolorbox}

\usepackage[nameinlink,noabbrev]{cleveref}
\crefname{theorem}{Theorem}{Theorems}
\Crefname{theorem}{Theorem}{Theorems}
\crefname{lemma}{Lemma}{Lemmas}
\Crefname{lemma}{Lemma}{Lemmas}
\crefname{proposition}{Proposition}{Propositions}
\Crefname{proposition}{Proposition}{Propositions}
\crefname{corollary}{Corollary}{Corollaries}
\Crefname{corollary}{Corollary}{Corollaries}
\crefname{definition}{Definition}{Definitions}
\Crefname{definition}{Definition}{Definitions}
\crefname{remark}{Remark}{Remarks}
\Crefname{remark}{Remark}{Remarks}
\crefname{assumption}{Assumption}{Assumptions}
\Crefname{assumption}{Assumption}{Assumptions}

\makeatletter
\def\qa@prefix@def{def}
\def\qa@prefix@lem{lem}
\def\qa@prefix@prop{prop}
\def\qa@prefix@cor{cor}
\def\qa@prefix@rem{rem}
\def\qa@prefix@ass{ass}
\def\qa@prefix@thm{thm}
\def\qa@prefix@sec{sec}
\def\qa@prefix@subsec{subsec}
\def\qa@prefix@subsubsec{subsubsec}
\def\qa@prefix@tab{tab}
\def\qa@prefix@fig{fig}
\def\qa@prefix@eq{eq}
\def\qa@prefix@app{app}
\def\qa@getlabelprefix#1:#2\@nil{\def\qa@labelprefix{#1}}
\newcommand{\qa@typedCref}[1]{%
  \def\qa@labelprefix{}%
  \expandafter\qa@getlabelprefix#1:\@nil
  \ifx\qa@labelprefix\qa@prefix@def
    Definition~\ref{#1}%
  \else\ifx\qa@labelprefix\qa@prefix@lem
    Lemma~\ref{#1}%
  \else\ifx\qa@labelprefix\qa@prefix@prop
    Proposition~\ref{#1}%
  \else\ifx\qa@labelprefix\qa@prefix@cor
    Corollary~\ref{#1}%
  \else\ifx\qa@labelprefix\qa@prefix@rem
    Remark~\ref{#1}%
  \else\ifx\qa@labelprefix\qa@prefix@ass
    Assumption~\ref{#1}%
  \else\ifx\qa@labelprefix\qa@prefix@thm
    Theorem~\ref{#1}%
  \else\ifx\qa@labelprefix\qa@prefix@sec
    Section~\ref{#1}%
  \else\ifx\qa@labelprefix\qa@prefix@subsec
    Section~\ref{#1}%
  \else\ifx\qa@labelprefix\qa@prefix@subsubsec
    Section~\ref{#1}%
  \else\ifx\qa@labelprefix\qa@prefix@tab
    Table~\ref{#1}%
  \else\ifx\qa@labelprefix\qa@prefix@fig
    Figure~\ref{#1}%
  \else\ifx\qa@labelprefix\qa@prefix@eq
    Equation~\ref{#1}%
  \else\ifx\qa@labelprefix\qa@prefix@app
    Appendix~\ref{#1}%
  \else
    \ref{#1}%
  \fi\fi\fi\fi\fi\fi\fi\fi\fi\fi\fi\fi\fi\fi
}
\newcommand{\qa@typedcref}[1]{%
  \def\qa@labelprefix{}%
  \expandafter\qa@getlabelprefix#1:\@nil
  \ifx\qa@labelprefix\qa@prefix@def
    definition~\ref{#1}%
  \else\ifx\qa@labelprefix\qa@prefix@lem
    lemma~\ref{#1}%
  \else\ifx\qa@labelprefix\qa@prefix@prop
    proposition~\ref{#1}%
  \else\ifx\qa@labelprefix\qa@prefix@cor
    corollary~\ref{#1}%
  \else\ifx\qa@labelprefix\qa@prefix@rem
    remark~\ref{#1}%
  \else\ifx\qa@labelprefix\qa@prefix@ass
    assumption~\ref{#1}%
  \else\ifx\qa@labelprefix\qa@prefix@thm
    theorem~\ref{#1}%
  \else\ifx\qa@labelprefix\qa@prefix@sec
    section~\ref{#1}%
  \else\ifx\qa@labelprefix\qa@prefix@subsec
    section~\ref{#1}%
  \else\ifx\qa@labelprefix\qa@prefix@subsubsec
    section~\ref{#1}%
  \else\ifx\qa@labelprefix\qa@prefix@tab
    table~\ref{#1}%
  \else\ifx\qa@labelprefix\qa@prefix@fig
    figure~\ref{#1}%
  \else\ifx\qa@labelprefix\qa@prefix@eq
    equation~\ref{#1}%
  \else\ifx\qa@labelprefix\qa@prefix@app
    appendix~\ref{#1}%
  \else
    \ref{#1}%
  \fi\fi\fi\fi\fi\fi\fi\fi\fi\fi\fi\fi\fi\fi
}
\DeclareRobustCommand{\Cref}[1]{\qa@typedCref{#1}}
\DeclareRobustCommand{\cref}[1]{\qa@typedcref{#1}}
\makeatother

\microtypesetup{expansion=false}
\newcolumntype{L}[1]{>{\raggedright\arraybackslash}p{#1}}
\newcolumntype{Y}{>{\raggedright\arraybackslash}X}

\theoremstyle{definition}
\newtheorem{theorem}{Theorem}
\newtheorem{lemma}[theorem]{Lemma}
\newtheorem{proposition}[theorem]{Proposition}
\newtheorem{corollary}[theorem]{Corollary}
\newtheorem{definition}[theorem]{Definition}
\newtheorem{assumption}[theorem]{Assumption}
\newtheorem{remark}[theorem]{Remark}

\definecolor{qaBlue}{RGB}{40,90,160}
\definecolor{qaLightBlue}{RGB}{220,232,246}
\definecolor{qaRed}{RGB}{180,50,50}
\definecolor{qaLightRed}{RGB}{248,224,224}
\definecolor{qaGreen}{RGB}{40,130,70}
\definecolor{qaLightGreen}{RGB}{224,242,228}
\definecolor{qaOrange}{RGB}{200,120,30}
\definecolor{qaLightOrange}{RGB}{250,236,214}
\definecolor{qaPurple}{RGB}{110,70,150}
\definecolor{qaLightPurple}{RGB}{234,226,246}
\definecolor{qaDark}{RGB}{40,40,40}
\definecolor{qaLightGrey}{RGB}{238,238,238}

\tikzset{
qaBox/.style={
rectangle,
rounded corners=2pt,
draw=qaDark!55,
line width=0.5pt,
minimum height=9mm,
align=center,
inner xsep=5pt
},
qaSmallBox/.style={
rectangle,
rounded corners=2pt,
draw=qaDark!55,
line width=0.5pt,
minimum height=7mm,
align=center,
inner xsep=4pt,
font=\scriptsize
},
qaArrow/.style={
-{Stealth[length=2.2mm]},
draw=qaDark!60,
line width=0.6pt
},
qaDashedArrow/.style={
-{Stealth[length=2.2mm]},
draw=qaDark!55,
line width=0.6pt,
dashed
},
qaLabel/.style={
font=\scriptsize\itshape,
text=qaDark!55
},
qaPanel/.style={
rectangle,
rounded corners=3pt,
draw=qaDark!25,
fill=qaLightGrey!35,
inner sep=5mm
},
}

\DeclareMathOperator{\supp}{supp}
\DeclareMathOperator{\polylog}{polylog}

\DeclareMathOperator{\rank}{rank}

\providecommand{\Irev}{\ensuremath{\mathcal{I}_{\mathrm{rev}}}}
\providecommand{\own}{\operatorname{own}}

\providecommand{\Unif}{\operatorname{Unif}}

\title{When Does a Quantum Speedup Survive End-to-End?\\
\large Interface Admissibility for Topological Data Analysis}

\author{
  Pablo Herrero G\'omez\thanks{ORCID: \href{https://orcid.org/0009-0003-7361-4900}{0009-0003-7361-4900}; \texttt{pablo.herrero@ua.es}} \and
  Antonio Jimeno Morenilla\thanks{ORCID: \href{https://orcid.org/0000-0002-3789-6475}{0000-0002-3789-6475}; \texttt{jimeno@ua.es}} \and
  David Muñoz Hernández\thanks{ORCID: \href{https://orcid.org/0009-0008-4218-2945}{0009-0008-4218-2945}; \texttt{david.mhernandez@ua.es}} \and
  Higinio Mora Mora\thanks{ORCID: \href{https://orcid.org/0000-0002-8591-0710}{0000-0002-8591-0710}; \texttt{hmora@ua.es}}\\
  \normalsize University of Alicante, Alicante, Spain
}
\date{}

\begin{document}

\maketitle

\begin{abstract}
Primitive quantum speedups are interface-relative: they depend on the input
access used to run the primitive and on the output contract used to consume its
state or samples. This paper introduces a transcript-level admissibility
relation \(A_M\preceq_{\mathrm{int}}A_Q\), defined relative to the declared
implementation package of the quantum interface. It identifies which adaptive
classical access transcripts that same package licenses, with all setup,
transcript-generation, and precision overheads charged.

The main application is an operational audit for normalized-Betti estimation in
clique-complex TDA, separating three declared-interface regimes. Reversible
indexed simplex interfaces certify matched classical simplex sampling and local
Laplacian row access by evaluating their reversible routines on single
computational branches. Membership-based preparations induce a rejection route
of overhead \(\binom{n}{k+1}/|S_k|\). Abstract spectral or block-encoding
interfaces require an accompanying implementation package, transcript reduction,
or shared representation. Under the indexed certificate and interface closure,
the end-to-end cost is fixed by the imported estimator's spectral dependence on
the gap \(\gamma\); the concretely realized bounded-treewidth family already
admits exact \(\mathrm{poly}(n)\) classical Betti computation by rank over
\(\mathbb{Q}\). A low-rank separation supports the role of access and output
contracts.

\end{abstract}


\section{Introduction}
\label{sec:introduction}

Primitive quantum speedups are statements relative to interfaces. A linear
algebra primitive may assume a block-encoding, QRAM-like loading, amplitude
access, or a state-preparation oracle. A TDA primitive may assume a simplex
register and a block-encoding of a combinatorial Laplacian. These assumptions
specify how the quantum primitive is run and what information about the
instance has already been made operationally available.

This paper asks when such a primitive speedup survives as applied end-to-end
advantage under explicit input and output interfaces. The answer is
interface-relative. If the declared quantum implementation also realizes a
matched classical route, an interface-closed baseline includes that route with
all declared overheads charged. Dense output contracts can also impose
extraction costs that shape end-to-end inheritance even when the primitive
remains fast.

The central object is the admissibility relation
\[
  A_M\preceq_{\mathrm{int}}A_Q .
\]
It says that a classical transducer, using only the declared implementation
package \(I_Q(x)\) for the quantum access model \(A_Q\), public parameters, and
its own randomness, can generate the adaptive transcripts of a classical access
model \(A_M\) with charged overhead. This is a condition on the declared
interface. For implemented packages, transcript reductions, or declared shared
representations, the relation supplies the charged classical sampler or access
transcript licensed by that interface.

The framework becomes operational through implementation-level certificates.
Once admissibility is certified, an interface-closed baseline includes the
matched classical route, and the classical end-to-end cost is at most the cost
of that route. The concrete work is to determine, for a declared implementation
package, which classical transcripts are generated by the same operational
resources as the quantum access.

For clique-complex TDA, \Cref{subsec:tda-audit} and
\Cref{thm:tda-interface-audit} give this determination as an audit of the
declared interface. An abstract spectral interface--a block-encoding,
state-preparation oracle, or spectral primitive--is classified as a spectral
declaration whose matched classical access is certified by an accompanying
implementation package, transcript reduction, or shared representation. A
membership-based simplex preparation induces a rejection route whose overhead
is the inverse clique-density factor \(\binom{n}{k+1}/|S_k|\). A reversible
indexed implementation, formalized in \Cref{def:rev-indexed-tda}, declares
reversible simplex-generation routines and reversible local transition/value
routines for the Laplacian; \Cref{thm:tda-indexed-collapse} derives the
matched classical simplex sampling and local-row transcripts by evaluating
those routines on computational-basis branches.

Combining the reversible indexed certificate with the imported estimator module
and interface closure places the matched classical route in the baseline at
overhead \(A_{\mathrm{idx}}=\operatorname{polylog}(N)\), so the end-to-end
comparison for additive normalized Betti estimation is fixed by the spectral
cost of the imported estimator in \(\gamma\). The realized indexed family of
bounded-treewidth complexes is classically tractable in \(\operatorname{poly}(n)\)
by exact rank computation, and \Cref{prop:tw-incompat} delineates the scope of
the indexed certificate. Exact Betti numbers, unnormalized Betti numbers,
homology-basis output, dense spectral output, abstract spectral declarations,
membership-only preparations with large clique-density overhead, and
hidden-algebraic regimes are evaluated by their own interface certificates or
independent classical lower bounds under the declared access model.

The remaining results support this interface view. A low-rank access separation
illustrates how entrywise access can face a localization barrier while
\(\ell_2\) sample-and-query access exposes the relevant mass directly. These
low-rank examples are used as supporting illustrations of access-relative
reducibility. The main applied certificate in this paper is the TDA interface
audit for normalized-Betti estimation. Output-contract lifting supplies the
complementary boundary: dense classical outputs, such as a full QLSA solution
vector, impose extraction costs beyond those of compact scalar-observable
contracts.

Section~\ref{sec:related} positions the paper relative to end-to-end
frameworks, dequantization, quantum TDA, and output-extraction bounds.
Section~\ref{sec:results} defines the interface-admissibility framework and
states the baseline-closure principle. Section~\ref{sec:applications} applies
the framework to normalized-Betti estimation in clique-complex TDA and
illustrates it on a low-rank access separation.
Appendix~\ref{sec:methods} contains the low-rank lower-bound proofs and
output-bound details.

\section{Related Work and Positioning}
\label{sec:related}

The literature on applied quantum advantage already treats access assumptions
as part of the computational claim. Analyses of quantum speedups emphasize that
data loading, oracle construction, hidden constants, finite precision, resource
regime, and output extraction can change an end-to-end comparison even when a
primitive subroutine is asymptotically fast \cite{aaronson_read_2015,
tang_quantum_2021}. Work on block-encodings and QRAM-like loading models makes
the same point from the input side: the declared interface is part of the
algorithmic object, not an implementation detail external to the cost model
\cite{giovannetti_quantum_2008,sunderhauf_block-encoding_2024}. The modelling
issue is especially visible in quantum machine learning, where data access and
state preparation are often stated as part of the task itself
\cite{biamonte_quantum_2017,schuld_supervised_2018}.

This paper is positioned inside that end-to-end viewpoint. It does not replace
resource estimates, oracle lower bounds, or application-specific cost models.
It adds a formal comparison rule for one recurring situation: when a classical
route is generated by the same declared implementation package that realizes
the quantum access, that route belongs in the interface-closed classical
baseline. The relation \(A_M\preceq_{\mathrm{int}}A_Q\) records this
transcript-level admissibility. Two claims may therefore share the same
primitive quantum subroutine while declaring different implementation packages,
and those packages can license different classical transcripts.

Quantum-inspired and dequantization results provide the closest technical
motivation for taking this interface distinction seriously. Sample-and-query
access can reproduce, or closely match, several low-rank quantum
linear-algebra primitives \cite{tang_quantum-inspired_2019,
gilyen_quantum-inspired_2018}. Later refinements sharpen the dependence on
rank, conditioning, sampling assumptions, and precision
\cite{chia_sampling-based_2022,bakshi_improved_2024}. These results operate on
declared sample-and-query access for low-rank structure. They show why an
applied comparison must say whether norm trees, sampling tables, unrank maps,
common representations, or reversible branch-evaluable routines are declared as
part of \(I_Q(x)\). If such objects are part of the declared quantum
implementation package, certain classical transcript models may factor through
that same package; when they do, the setup, preprocessing, storage, sampling,
arithmetic, precision conversion, and amplification costs are charged and the
route enters the interface-closed baseline. If the claim declares only an
abstract oracle, state-preparation primitive, or block-encoding, then the
abstract interface itself is what has been certified; additional transcript
data are needed before sample-and-query, entrywise, row, or sampling access can
be treated as matched classical access.

Topological data analysis is the main later application of this framework.
Quantum algorithms for Betti-number estimation typically organize the input
around simplex registers, membership predicates, boundary or Laplacian
operators, block-encodings, phase-estimation routines, or spectral
transformations \cite{lloyd_quantum_2016,mcardle_streamlined_2022}. Complexity
results for homological tasks and clique-complex formulations likewise make the
access model and the requested output central to the interpretation of the
claim \cite{gyurik_towards_2020,schmidhuber_complexity_2022}. The later TDA
application uses classical normalized-Betti and spectral-trace estimators as
imported modules. Path-integral Monte Carlo and related stochastic
trace-estimation approaches provide classical routines under appropriate
sampling and local-row access \cite{apers_simple_2022,akhalwaya_comparing_2024}.
The role of the framework is to determine when that access is induced by the
declared quantum implementation package and therefore belongs to the
interface-closed baseline.

Output contracts form the complementary side of the end-to-end comparison.
Producing a dense classical vector, a full distribution, or a tomographic
description is operationally different from estimating a scalar observable or
generating samples. The QLSA/HHL family is the standard calibration point: a
primitive that prepares a quantum state can have compact cost for scalar
observables while a dense classical description of the state carries extraction
or write-out costs \cite{harrow_quantum_2009,aaronson_read_2015}. Tomography
and low-rank tomography impose dimension-dependent sample requirements under
standard reconstruction contracts \cite{haah_sample-optimal_2017,
yuen_improved_2023}, while classical-shadow bounds apply to suitable
observable families rather than dense reconstruction
\cite{huang_predicting_2020}. The framework below keeps this distinction
explicit through the output contract \(O\).

The original material in this manuscript is the interface-admissibility layer
used to audit end-to-end comparisons. It defines a transcript-level relation
\(A_M\preceq_{\mathrm{int}}A_Q\), defines interface-closed baselines that
include charged matched classical routes, and gives an audit recipe for applied
claims: identify the declared implementation package, derive the matched
transcript models that factor through it, charge their overheads, and compare
the resulting route with the declared quantum end-to-end cost. Later sections
apply this recipe to normalized-Betti estimation in clique-complex TDA. Table
\ref{tab:framework-positioning} summarizes the resulting division between the
paper's interface layer and the estimator, dequantization, and readout modules
that it imports.

\begin{table}[t]
\centering
\small
\setlength{\tabcolsep}{3pt}
\begin{tabularx}{\textwidth}{@{}L{3.2cm}L{2.6cm}Y@{}}
\toprule
Component & Status here & Role in the paper \\
\midrule
Interface admissibility
\(A_M\preceq_{\mathrm{int}}A_Q\) &
Original definition &
Identifies classical transcript models generated by the same declared package
that realizes the quantum access. \\
Interface-closed baselines &
Original comparison rule &
Includes charged matched classical routes in the baseline infimum. \\
Audit recipe &
Original framework component &
Turns an applied claim into a check of declared implementation package,
matched transcripts, charged overhead, and end-to-end competitiveness. \\
Normalized-Betti estimators &
Imported modules &
Used later under the sampling and local-row access conditions supplied by the
classical TDA literature. \\
Quantum-inspired low-rank routines &
Imported modules &
Used when the relevant sample-and-query access is certified as matched. \\
Readout and tomography bounds &
Imported modules &
Calibrate the distinction between compact scalar contracts and dense output
contracts. \\
\bottomrule
\end{tabularx}
\caption{Narrative map of the framework. The table separates the original
interface-admissibility layer from estimator, dequantization, and readout
modules imported from the literature.}
\label{tab:framework-positioning}
\label{tab:proved-imported}
\end{table}

\section{Interface-Admissibility Framework}
\label{sec:results}

\subsection{Applied claims and output contracts}
\label{subsec:contract-model}

\begin{definition}[Applied claim]
\label{def:applied-claim}
For each nominal size $N$, an applied quantum-advantage claim is a tuple
\[
\mathsf{Claim}_N=(\mathcal C_N,A_Q,A_C,O,\mathcal B).
\]
Here $\mathcal C_N$ is a class of instances, $A_Q$ is the quantum access model,
$A_C$ is the declared classical access standard, $O$ is an output contract, and
$\mathcal B$ is a family of admissible classical baselines. The quantum
primitive, resource regime, and scaling convention are declared with the claim
when they enter the cost model, but they are not separate components of the
typed access/output comparison.
\end{definition}

\begin{definition}[Output contract]
\label{def:output-contract}
An output contract is a triple
\[
O=(\mathcal Y_N,\tau_N,\Delta_N),
\]
where $\mathcal Y_N$ is the response space, $\tau_N:\mathcal C_N\to\mathcal T_N$
is the ideal output map, and
\[
\Delta_N:\mathcal Y_N\times\mathcal T_N\to[0,\infty]
\]
is an error functional. An algorithm satisfies $O$ with accuracy $\varepsilon$
and failure probability $\delta$ if, for every $x\in\mathcal C_N$, its output
$Y$ obeys
\[
\Pr[\Delta_N(Y,\tau_N(x))\le \varepsilon]\ge 1-\delta.
\]
\end{definition}

\begin{definition}[Compact and dense output contracts]
\label{def:compact-dense-output}
An output contract is \(N\)-compact in a declared regime if it can be satisfied
using
\[
N_{\mathrm{shots}}(O,N,\varepsilon,\delta)
\le
\operatorname{polylog}(N)
\operatorname{poly}(1/\varepsilon,\log(1/\delta))
\]
measurement transcripts, samples, or repetitions, together with postprocessing
bounded by the same kind of factors. A contract is \(N\)-dense if satisfying it
requires producing or reconstructing \(N^{\Omega(1)}\) classical degrees of
freedom under the declared metric. Dense output requirements are charged as
part of the end-to-end comparison.
\end{definition}

The quantum end-to-end cost is written as
\begin{equation}
\label{eq:quantum-e2e}
T_Q^{\mathrm{E2E}}(N,\varepsilon,\delta)
=
\rho_R(N,\varepsilon)
N_{\mathrm{shots}}(O,N,\varepsilon,\delta)
T_{\mathrm{cycle}}(N,\varepsilon),
\end{equation}
where
\begin{equation}
\label{eq:cycle-cost}
T_{\mathrm{cycle}}(N,\varepsilon)
=
T_{A_Q}(N,\varepsilon)
+
\alpha_{A_Q}(N,\varepsilon)T_Q^{\mathrm{prim}}(N,\varepsilon)
+
T_{\mathrm{meas}}(N,\varepsilon).
\end{equation}
The interface-closed classical comparison is
\begin{equation}
\label{eq:classical-e2e}
T_C^{\mathrm{E2E}}(N,\varepsilon,\delta)
=
\inf_{B\in\mathcal B_{\mathrm{int}}}
T_B^{A_C^{\mathrm{eff}}(A_C,A_Q),O}(\mathcal C_N,\varepsilon,\delta),
\end{equation}
where \(A_C^{\mathrm{eff}}\) and \(\mathcal B_{\mathrm{int}}\) are made
explicit below. Before closure, \(A_C\) denotes the declared classical access
standard; after closure, the baseline cost is evaluated over the effective
closed access family. The cost of a baseline includes its setup, access,
reduction, and final operation costs under the relevant access standard.

\subsection{Transcript-level interface admissibility}
\label{subsec:access-relative-reducibility}

\begin{definition}[Interface-matched access]
\label{def:interface-matched-access}
Let $\mathcal C_N$ be an instance class and let $x\in\mathcal C_N$. Let $A_Q$
be an abstract quantum access model, and let $I_Q(x)$ denote the declared
implementation package used by the applied claim to realize $A_Q$ on $x$. This
package may be a data structure, loading routine, oracle construction,
state-preparation circuit, unrank or sampling routine, or declared common
representation. Let $A_M$ be a candidate classical access model, with ideal
adaptive transcript distribution for $A_M(x)$.

We write
\[
A_M\preceq_{\mathrm{int}} A_Q,
\]
with overhead profile $\Gamma_M$, if there exists a classical transducer $T_M$
such that, for every $x\in\mathcal C_N$, $T_M$ uses only $I_Q(x)$, public
parameters, and its own randomness to generate valid adaptive transcripts of
$A_M(x)$. The charged cost $\Gamma_M$ includes setup, preprocessing, storage,
query and response generation, sampling, arithmetic, precision conversion,
verification, and failure amplification.

The transducer obeys an information restriction: its transcript distribution
must factor through \(I_Q(x)\). If two instances induce the same declared
implementation package, or packages indistinguishable within the declared
precision, then the generated \(A_M\)-transcript distributions must be
identical, or indistinguishable within the corresponding transcript precision.
The transducer may not use raw-instance access, hidden advice, uncharged
preprocessing, or any stronger representation of the instance than the one
declared to realize \(A_Q\).

Approximate transcript generation is treated by charging the transcript
precision conversion and failure amplification into $\Gamma_M$; unless another
metric is declared, transcript distance means total variation distance over the
full adaptive transcript.

The relation is certified in three ways.
\begin{enumerate}[label=(\alph*)]
\item \label{case:match-abstract-oracle}\emph{Abstract-oracle case.} An abstract
quantum oracle, abstract state-preparation oracle, or abstract block-encoding
is an abstract declaration. A matched access certificate is supplied by an
implementation package, transcript reduction, or shared representation with
charged overhead.
\item \label{case:match-data-structure}\label{case:match-oracle-transcript}
\emph{Implemented-interface case.} If the claim declares an implementation
package, matching is certified by giving a charged transducer from that package
to $A_M$-transcripts.
\item \label{case:match-declared-common}\emph{Declared-common-model case.} If
the applied model grants a common interface $G(x)$ from which both $A_Q$ and
$A_M$ are obtained with charged costs, then matching is certified through
$G(x)$.
\end{enumerate}
For example, a norm tree used for state preparation realizes
\(\ell_2\) sample-and-query access, and an unrank map for a uniform simplex
register realizes a classical uniform simplex sampler. Implemented
block-encodings certify sample-and-query access, row access, entry access, or
uniform sampling through the declared implementation, transcript reduction, or
shared data structure that generates the corresponding transcripts.
\end{definition}

\begin{remark}[Interface matching as admissibility]
\label{rem:interface-matching-admissibility}
Interface matching is an admissibility relation for baseline comparisons. It
states when the declared input interface of the quantum algorithm licenses a
charged classical route using \(A_M\).
\end{remark}

\subsection{Interface-closed baselines}
\label{subsec:interface-closed-baselines}

\begin{definition}[Interface-closed baseline family]
\label{def:interface-closed-baseline}
The baseline family $\mathcal B$ is interface-closed for the claim
$\mathsf{Claim}_N$ if, for every $A_M\preceq_{\mathrm{int}}A_Q$, its closed
version $\mathcal B_{\mathrm{int}}$ contains the classical algorithms that use
$A_M$ whenever their setup, access, reduction, and final operation costs are
charged. An interface-closed comparison includes the interface-matched
classical routes whose costs are charged in this sense.
\end{definition}

\paragraph{Effective closed access family.}
\label{def:effective-classical-access}
Identifying a single declared access model with the corresponding singleton
family, define
\[
A_C^{\mathrm{eff}}(A_C,A_Q)
:=
A_C\cup
\{\,A_M:\; A_M\preceq_{\mathrm{int}}A_Q\,\}.
\]
Here $A_C$ is the declared pre-closure classical access standard, while
$A_C^{\mathrm{eff}}(A_C,A_Q)$ is the effective family of access models admitted
in an interface-closed comparison. In such comparisons, occurrences of $A_C$ in
baseline costs are understood as this effective closed access family, and
$\mathcal B_{\mathrm{int}}$ denotes the baseline family closed under those
admissible accesses with all costs charged.

\paragraph{Implementation packages determine admissibility.}
The same primitive-level quantum behavior can receive different end-to-end
verdicts depending on the declared implementation package. A package containing
a norm tree, sampler, unrank map, or common representation can license concrete
classical transcripts, and the admissible baseline is determined by the
transcripts generated by the declared package.

\begin{lemma}[Interface-closure principle]
\label{lem:interface-closure-principle}
\label{thm:matched-interface-collapse}
\label{prop:closure-consequence}
Let
\[
\mathsf{Claim}_N=(\mathcal C_N,A_Q,A_C,O,\mathcal B)
\]
be an applied claim, and let \(\mathcal B_{\mathrm{int}}\) be
interface-closed for this claim. Suppose that
\(A_M\preceq_{\mathrm{int}}A_Q\) with charged overhead \(\Gamma_M\), and that a
classical route using \(A_M\) satisfies the output contract \(O\) with total
charged cost
\[
T_M(N,\varepsilon,\delta).
\]
Then
\[
T_C^{\mathrm{E2E}}(N,\varepsilon,\delta)
\le
T_M(N,\varepsilon,\delta).
\]
Consequently, if
\[
T_M(N,\varepsilon,\delta)
\le
T_Q^{\mathrm{E2E}}(N,\varepsilon,\delta)
\operatorname{polylog}(N)
\operatorname{poly}(1/\varepsilon,\log(1/\delta)),
\]
then the interface-closed comparison places the classical end-to-end cost
within the displayed polylogarithmic and precision factor of the quantum
end-to-end cost in \(N\) for that claim.
\end{lemma}

\begin{proof}
By interface closure, the matched classical route is an admissible member of
the baseline comparison once its setup, access, reduction, and operation costs
are charged. Since $T_C^{\mathrm{E2E}}$ is an infimum over
$\mathcal B_{\mathrm{int}}$ with access
$A_C^{\mathrm{eff}}(A_C,A_Q)$, it is at most the cost $T_M$ of this
route. If the displayed competitiveness condition on \(T_M\) holds, the ratio
\(T_C^{\mathrm{E2E}}/T_Q^{\mathrm{E2E}}\) is bounded by a polylogarithmic
factor in \(N\) times polynomial precision and amplification factors.
\end{proof}

This principle is a comparison rule under declared interfaces: its use requires
an explicit admissibility certificate \(A_M\preceq_{\mathrm{int}}A_Q\), an
explicit accounting of \(\Gamma_M\), and an interface-closed baseline.

\subsection{Audit recipe}
\label{subsec:audit-recipe}

Given an applied claim
\(\mathsf{Claim}_N=(\mathcal C_N,A_Q,A_C,O,\mathcal B)\), the audit proceeds as
follows.
\begin{enumerate}[label=\textbf{Step \arabic*.}]
\item Identify the declared quantum implementation package \(I_Q(x)\).
\item Determine which classical transcript models \(A_M\) factor through
\(I_Q(x)\), public parameters, and private randomness.
\item Charge \(\Gamma_M\) explicitly, including setup, preprocessing, storage,
query generation, precision conversion, arithmetic, and amplification.
\item Close the classical baseline over the effective access family
\[
A_C^{\mathrm{eff}}(A_C,A_Q)
=
A_C\cup\{A_M:A_M\preceq_{\mathrm{int}}A_Q\}.
\]
\item Check whether the resulting matched route is competitive with
\(T_Q^{\mathrm{E2E}}\) under the requested output contract.
\item If the route is not competitive, record the responsible parameter,
missing interface component, or output requirement.
\end{enumerate}

The later TDA application applies this recipe to simplex generation, local
Laplacian transcripts, clique-density overheads, and normalized-Betti output
contracts.

\section{Main Applications of the Framework}
\label{sec:applications}

\Cref{sec:results} gave a general audit framework for applied claims: specify
the output contract, identify the declared quantum implementation package,
derive the classical transcript models that factor through that package, and
close the classical baseline over those matched accesses with all overheads
charged. This section applies the framework to concrete settings. The first
application is the main one: additive normalized Betti estimation for
clique-complex topological data analysis. The audit decides which classical
transcript models are licensed by the declared quantum TDA implementation
package.

The TDA audit separates three interface classes: abstract spectral or
block-encoding declarations; membership-based simplex preparations; and
reversible indexed simplex/Laplacian implementations. These classes lead to
different interface-closed baselines. Section~\ref{sec:lowrank} illustrates the
same access-relative principle on a purely low-rank family.

\subsection{Application I: TDA Interface Audit for Normalized Betti Estimation}
\label{subsec:tda-audit}
\label{subsec:tda-indexed-collapse}

\subsubsection{Task, promise, and output contract}
\label{subsubsec:tda-task-contract}

Let \(G\) be a graph on \(n\) vertices and let \(X(G)\) be its clique complex.
For a fixed dimension \(k\), let \(S_k\) denote the set of \(k\)-simplices and
write \(N=|S_k|\). We encode a \(k\)-simplex as a sorted \((k+1)\)-tuple of
vertices. Let \(\Delta_k\succeq 0\) be the \(k\)-th combinatorial Laplacian,
with the normalization convention declared by the algorithm or estimator.

\begin{definition}[Typed normalized-Betti TDA claim]
\label{def:tda-class}
For parameters \(N,\gamma,\Lambda\) with \(0<\gamma\le\Lambda\), let
\[
  \mathcal{C}^{\mathrm{TDA}}_{N,\gamma,\Lambda}
  =
  \bigl\{\,X(G)\ \text{clique complex}:\ |S_k|=N,\ 
  \operatorname{spec}(\Delta_k)\subseteq\{0\}\cup[\gamma,\Lambda]\,\bigr\}.
\]
The output contract is additive normalized Betti estimation: output
\(\widehat{\beta}\in\mathbb R\) such that
\[
\Pr\!\left[
\left|\widehat{\beta}-\frac{\beta_k}{|S_k|}\right|
\le \varepsilon
\right]\ge 1-\delta .
\]
Equivalently, in the notation of \Cref{def:output-contract}, the ideal output
map is \(\tau_N(X(G))=\beta_k/|S_k|\), and the error functional is
\[
\Delta_N(\widehat{\beta},\beta_k/|S_k|)
=
\left|\widehat{\beta}-\frac{\beta_k}{|S_k|}\right|.
\]
The spectral promise is part of the typed applied claim and is recorded as an
explicit declared condition.
\end{definition}

The advantage scale for clique-complex TDA is the vertex count \(n\) together
with the spectral resolution \(1/\sqrt{\gamma}\). With \(|S_k|=N\), a
\(\operatorname{poly}(N)\) bound equals a \(\operatorname{poly}(n)\) bound when
\(N=\operatorname{poly}(n)\), and the two scales separate when
\(N=2^{\Theta(n)}\), where advantage means a quantum routine in
\(\operatorname{poly}(n)\) against a classical cost of \(n^{\omega(1)}\).

\subsubsection{The no-free-simplex-register principle}
\label{subsubsec:no-free-simplex-register}

\begin{theorem}[TDA interface trichotomy for normalized Betti estimation]
\label{thm:tda-interface-trichotomy}
\label{thm:tda-interface-audit}
\label{cor:tda-interface-audit}
For additive normalized Betti estimation on clique complexes under
\Cref{def:tda-class}, the declared quantum TDA implementation package
determines the following interface-audit outcomes.

\begin{enumerate}[label=(\roman*)]
\item \emph{Abstract spectral/block-encoding interface.}
An abstract block-encoding, state-preparation oracle, simplex-register oracle,
or spectral primitive for \(\Delta_k\) is a spectral declaration. Matched
classical simplex sampling, local Laplacian row access, and normalization
access are certified when the declaration includes an implementation package,
transcript reduction, or shared representation generating those transcripts
with charged overhead.

\item \emph{Membership-based simplex interface.}
If the implementation prepares the simplex register by starting from the
uniform distribution or superposition over all \((k+1)\)-subsets of \([n]\) and
applying a reversible clique-membership predicate, then the matched classical
route induced by that predicate is rejection sampling. Its expected overhead
per accepted simplex is
\[
R_k(G)=\frac{\binom{n}{k+1}}{|S_k|}.
\]
This route is polylogarithmic in \(N=|S_k|\) only when \(R_k(G)\) is
polylogarithmic, unless additional generative structure is declared and
charged.

\item \emph{Reversible indexed simplex/Laplacian interface.}
If the implementation declares a reversible indexed simplex generator, local
reversible transition/value routines for \(\Delta_k\), and
normalization/counting information for \(N=|S_k|\), then single-branch
evaluation certifies a matched classical access model consisting of
near-uniform simplex sampling, local row transcripts of \(\Delta_k\), and
normalization information. In this case
\[
A^{\mathrm{TDA}}_{\mathrm{simp,row},N}
\preceq_{\mathrm{int}} A_Q .
\]
\end{enumerate}
\end{theorem}

\begin{proof}
Part (i) is the abstract-oracle case of
\Cref{def:interface-matched-access}: the matched classical transcript
distribution for simplex sampling, local rows, or normalization is certified by
the declared implementation package, transcript reduction, or shared
representation through which it factors. Part (ii) follows from
\Cref{prop:clique-density-obstruction}. The declared membership
predicate licenses rejection sampling from the ambient set of
\((k+1)\)-subsets, with acceptance probability
\(|S_k|/\binom{n}{k+1}\). Part (iii) follows from
\Cref{thm:tda-indexed-collapse}: evaluating the declared reversible
simplex-generation and local Laplacian routines on computational-basis
branches generates the matched sampling, row, and counting transcripts with all
overheads charged.
\end{proof}

\paragraph{Abstract-interface clause.}
The first branch is the information restriction in
\Cref{def:interface-matched-access}: an abstract block-encoding or
state-preparation oracle certifies the declared abstract interface, and a
matched sampler or local-row transcript is certified by the additional declared
implementation information that generates the corresponding transcript.

\subsubsection{Reversible indexed implementation and matched local access}
\label{subsubsec:rev-indexed-support}

\begin{definition}[Reversible indexed TDA implementation]
\label{def:rev-indexed-tda}
A declared TDA implementation package \(I_Q(x)\) is a
\emph{reversible combinatorial indexed implementation}, written
\(I_Q(x)\in\mathcal I_{\mathrm{rev}}\), if it realizes the declared quantum
access \(A_Q\) and contains the following charged, gate-level objects.

\begin{enumerate}[label=(\roman*)]
\item The integer \(N=|S_k|\), or a declared counting/normalization routine with
charged cost \(T_{\mathrm{count}}\).

\item A reversible classical circuit \(U_{\mathrm{gen}}\), of size
\(T_{\mathrm{gen}}\), used by the quantum implementation to generate the
\(k\)-simplex register. On computational-basis inputs it implements a map
\[
G_{\mathrm{simp}}\colon \mathcal R\to S_k
\]
from a declared finite seed space \(\mathcal R\) to \(k\)-simplices. In the
exact indexed case, \(\mathcal R=[N]\) and \(G_{\mathrm{simp}}\) is a bijection
\([N]\to S_k\). More generally, when \(r\sim \mathrm{Unif}(\mathcal R)\), the
output distribution of \(G_{\mathrm{simp}}(r)\) is declared to be
\(\eta_{\mathrm{samp}}\)-close in total variation distance to uniform over
\(S_k\).

\item Reversible classical circuits \(O_{\mathrm{nb}}\) and \(O_{\mathrm{val}}\),
used by the declared Laplacian block-encoding, which generate local transition
data for \(\Delta_k\) on computational-basis inputs. Given a simplex \(\sigma\)
and an index \(\ell\le d_\Delta\), \(O_{\mathrm{nb}}\) returns the \(\ell\)-th
declared local neighbor candidate \(\tau_\ell(\sigma)\). Given
\((\sigma,\tau)\), \(O_{\mathrm{val}}\) returns an approximation to
\((\Delta_k)_{\sigma\tau}\) to \(b\) bits. The number of declared local neighbor
candidates is bounded by \(d_\Delta\), and the charged per-neighbor evaluation
cost is \(T_{\mathrm{row}}\).

\item The public parameters \(n,k\), the declared graph-access representation
used by the quantum implementation, and all arithmetic and precision conventions
needed to evaluate the above routines.
\end{enumerate}

The definition grants the reversible routines used coherently by
the declared quantum implementation. The content of the theorem below is that
single-branch evaluation of those routines induces the matched classical access.
\end{definition}

\begin{lemma}[Single-branch evaluation]
\label{lem:single-branch}
Let \(C\) be a circuit composed of reversible classical gates acting on
computational-basis states. Then \(C\) induces a deterministic map on basis
strings, and its value on any basis input can be computed classically by
simulating the gates once, with time \(O(|C|)\) up to the declared arithmetic
cost.
\end{lemma}

\begin{proof}
Each reversible classical gate maps computational-basis states to
computational-basis states. Their composition therefore defines a deterministic
map on basis strings. Evaluating the gates sequentially on one input string
computes that map with one classical pass through the circuit.
\end{proof}

\begin{lemma}[Smoothed zero-indicator trace under a spectral promise]
\label{lem:smoothed-zero-indicator-betti}
Let \(\Delta_k\succeq 0\) and suppose
\[
\operatorname{spec}(\Delta_k)\subseteq \{0\}\cup[\gamma,\Lambda].
\]
Let \(f_\gamma\) be a bounded smoothed zero-indicator satisfying
\(f_\gamma(0)=1\) and
\[
|f_\gamma(\lambda)|\le \eta
\qquad\text{for all }\lambda\in[\gamma,\Lambda].
\]
Then
\[
\left|
\frac{1}{|S_k|}\operatorname{tr} f_\gamma(\Delta_k)
-
\frac{\beta_k}{|S_k|}
\right|
\le \eta .
\]
\end{lemma}

\begin{proof}
The zero eigenvalues of \(\Delta_k\) contribute exactly
\(\beta_k f_\gamma(0)=\beta_k\) to the trace. Each nonzero eigenvalue lies in
\([\gamma,\Lambda]\) and contributes at most \(\eta\) in absolute value. After
normalizing by \(|S_k|\), the nonzero contribution is at most
\(\eta(|S_k|-\beta_k)/|S_k|\le \eta\).
\end{proof}

\begin{theorem}[Indexed TDA matched-access certificate]
\label{thm:tda-indexed-collapse}
Let \(X(G)\) be a clique complex, let \(S_k\) be its set of \(k\)-simplices, and
let \(N=|S_k|\). Suppose the declared quantum TDA access \(A_Q\) is realized by
a package \(I_Q(x)\in\mathcal I_{\mathrm{rev}}\) in the sense of
\Cref{def:rev-indexed-tda}. Let \(A_M\) be the classical access model
consisting of:
\[
\text{sampling }\sigma\in S_k,\qquad
\text{local row transcripts of }\Delta_k,\qquad
\text{and normalization/counting information for }N.
\]
Then
\[
A_M\preceq_{\mathrm{int}} A_Q .
\]
More precisely, there exists a classical transducer \(T_M\), using only
\(I_Q(x)\), public parameters, and its own randomness, whose sampling marginal
is \(\eta_{\mathrm{samp}}\)-close in total variation distance to uniform on
\(S_k\), and whose row answers have the declared \(b\)-bit precision. The
charged overhead for one sampled simplex and one local row transcript is
\[
A_{\mathrm{idx}}(N,k,d_\Delta,b)
=
O\!\left(
T_{\mathrm{gen}}
+
d_\Delta T_{\mathrm{row}}
+
T_{\mathrm{count}}
+
\operatorname{polylog}(1/\eta_{\mathrm{samp}})
+
b
\right).
\]
In the exact unranking case, where \(G_{\mathrm{simp}}\colon[N]\to S_k\) is a
bijection and the transducer samples \(i\sim\mathrm{Unif}[N]\), the sampling
error is \(\eta_{\mathrm{samp}}=0\).
\end{theorem}

This certificate identifies the classical transcripts that the declared package
\(I_Q(x)\in\mathcal I_{\mathrm{rev}}\) licenses; single-branch evaluation of
the declared routines (\Cref{lem:single-branch}) produces them with charged
overhead \(A_{\mathrm{idx}}\).

\begin{proof}
The transducer first generates the sampling part of the transcript. In the
exact indexed case, it samples \(i\sim\mathrm{Unif}[N]\) using its own
randomness and evaluates \(U_{\mathrm{gen}}\) on the computational-basis input
\(i\). By \Cref{lem:single-branch}, this costs \(T_{\mathrm{gen}}\). Since
\(G_{\mathrm{simp}}\) is a bijection from \([N]\) to \(S_k\), the resulting
simplex is exactly uniform over \(S_k\). In the approximate generative case,
the same branch evaluation gives the distribution declared for
\(G_{\mathrm{simp}}\), and the total variation error \(\eta_{\mathrm{samp}}\) is
charged as part of the implementation package.

For a local row query at a simplex \(\sigma\), the transducer evaluates the
declared local transition primitives on computational-basis inputs. For each
\(\ell\le d_\Delta\), it computes
\[
\tau_\ell(\sigma)=O_{\mathrm{nb}}(\sigma,\ell)
\]
and then computes the corresponding value
\[
O_{\mathrm{val}}(\sigma,\tau_\ell)
\]
to \(b\) bits. These are the same reversible classical subroutines used
coherently inside the declared Laplacian block-encoding; evaluating them on one
basis input produces a classical local-row transcript. The charged cost is
\(O(d_\Delta T_{\mathrm{row}}+b)\), with all finite-precision conversion
included.

The normalization information is obtained by reading the declared integer
\(N\), or by running the declared counting/normalization routine at charged cost
\(T_{\mathrm{count}}\). The theorem applies to implementations that declare
this normalization object and charge its cost.

The transducer is valid for adaptive transcripts because every future answer is
produced by a fresh evaluation of the same deterministic routines on the queried
basis input, together with fresh private randomness for new samples. Hence the
transcript distribution factors through \(I_Q(x)\). If two instances induce the
same declared implementation package up to the declared precision, the resulting
transcript distributions are identical up to that precision. The transducer
uses no raw-instance access, hidden advice, or uncharged preprocessing beyond
the declared implementation package. This is exactly the admissibility condition
\(A_M\preceq_{\mathrm{int}}A_Q\) from
\Cref{def:interface-matched-access}.
\end{proof}

\subsubsection{Imported normalized-Betti estimator and conditional competitiveness consequence}
\label{subsubsec:imported-betti-nogap}

\begin{assumption}[Imported normalized-Betti estimator module]
\label{ass:imported-betti-estimator}
\label{thm:imported-betti-estimator}
Let \(X(G)\) be a clique complex on \(n\) vertices, let \(S_k\) be the set of
\(k\)-simplices, and let \(N=|S_k|\). Let \(\Delta_k\succeq 0\) be the
\(k\)-th combinatorial Laplacian, normalized according to the convention used
by the estimator, and assume the spectral promise
\[
    \operatorname{spec}(\Delta_k)\subseteq \{0\}\cup[\gamma,\Lambda].
\]
Suppose an access model provides the following operations:
\begin{enumerate}[label=(\roman*)]
\item samples from a distribution on \(S_k\) that is
\(\eta_{\mathrm{samp}}\)-close in total variation distance to the uniform
distribution on \(S_k\);
\item local row transcripts for \(\Delta_k\), with at most \(d_\Delta\) declared
local neighbors per row and \(b\)-bit value precision;
\item the normalization information \(N=|S_k|\), or an equivalent charged
normalization routine.
\end{enumerate}
The classical normalized-Betti estimation literature supplies an estimator for
\[
    \frac{\beta_k}{|S_k|}
\]
with additive error \(\varepsilon\) and failure probability at most \(\delta\).
We denote its access and arithmetic operation count by
\[
    M_{\mathrm{Betti}}^{\mathrm{imp}}
    =
    M_{\mathrm{Betti}}^{\mathrm{imp}}
    (n,N,k,d_\Delta,\gamma,\varepsilon,\delta,b,\eta_{\mathrm{samp}}).
\]

In particular, one concrete imported instantiation is the path-integral Monte
Carlo estimator of Apers--Gribling--Sen--Szabo. For general simplicial
complexes, its running time is bounded by
\[
    n^{\,O\!\left(\frac{1}{\sqrt{\gamma}}
    \log\frac{1}{\varepsilon}\right)}
\]
up to polynomial factors in the remaining precision, amplification, and local
access costs. For clique complexes, their improved bound can be written as
\[
    \left(\frac{n}{\lambda_{\max}}\right)^{
    O\!\left(\frac{1}{\sqrt{\gamma}}\log\frac{1}{\varepsilon}\right)}
\]
under the normalization and eigenvalue conventions of that estimator, again
with the corresponding precision, amplification, and local-access costs
charged. The Monte Carlo comparison framework of Akhalwaya et al. provides an
additional clique-complex BNE module with explicit sample-complexity bounds for
the corresponding stochastic trace-estimation subroutines.

All polynomial approximation, variance, local-walk, finite-precision, sampling
bias, and amplification costs of the imported estimator are included in
\(M_{\mathrm{Betti}}^{\mathrm{imp}}\). The present paper uses these results as
imported estimator modules; its contribution is to determine when the local
sampling and row-access requirements of such modules are licensed by the
declared quantum implementation interface.
\end{assumption}

\begin{remark}[Parameter dependence of the imported estimator]
\label{rem:betti-parameter-dependence}
The quantity \(M_{\mathrm{Betti}}^{\mathrm{imp}}\) is intentionally kept as a
parameterized imported module in the interface theorem. The matched-access
audit records the estimator's declared dependence on the ambient parameters.
The end-to-end consequence depends on the declared regime: the relevant
parameters include the ambient vertex count \(n\), the simplex
count \(N=|S_k|\), the dimension \(k\), the local row bound \(d_\Delta\), the
spectral resolution \(1/\gamma\), the requested precision \(1/\varepsilon\),
the failure amplification \(\log(1/\delta)\), the value precision \(b\), and
the sampling error \(\eta_{\mathrm{samp}}\). The matched classical route is
competitive in the declared regimes where these quantities have the stated
polylogarithmic or polynomial dependence on \(N\).
\end{remark}

\begin{corollary}[Conditional matched-interface competitiveness criterion for normalized Betti estimation]
\label{cor:tda-conditional-nogap}
\label{cor:tda-normalized-betti-bound}
Assume \Cref{thm:tda-indexed-collapse},
\Cref{ass:imported-betti-estimator}, and an interface-closed baseline family.
Then the interface-closed classical end-to-end cost satisfies
\[
T^{\mathrm{E2E}}_C(N,\varepsilon,\delta)
\leq
A_{\mathrm{idx}}(N,k,d_\Delta,b)\,
M_{\mathrm{Betti}}^{\mathrm{imp}}
(n,N,k,d_\Delta,\gamma,\varepsilon,\delta,b,\eta_{\mathrm{samp}})
+
T_{\mathrm{setup}} .
\]
Consequently, in any declared regime where
\[
A_{\mathrm{idx}}(N,k,d_\Delta,b)
\,
M_{\mathrm{Betti}}^{\mathrm{imp}}
(n,N,k,d_\Delta,\gamma,\varepsilon,\delta,b,\eta_{\mathrm{samp}})
+
T_{\mathrm{setup}}
\leq
T^{\mathrm{E2E}}_Q(N,\varepsilon,\delta)\,
\operatorname{polylog}(N)\operatorname{poly}(1/\varepsilon,\log(1/\delta)),
\]
the interface-closed comparison places the classical end-to-end cost within the
displayed polylogarithmic and precision factor of the quantum end-to-end cost in
\(N=|S_k|\).
\end{corollary}

\begin{proof}
\Cref{thm:tda-indexed-collapse} certifies that the reversible indexed
implementation induces the matched classical access model consisting of
simplex sampling, local row transcripts of \(\Delta_k\), and normalization
information. \Cref{ass:imported-betti-estimator} supplies an imported
estimator under exactly that access model, with all estimator-side costs
included in \(M_{\mathrm{Betti}}^{\mathrm{imp}}\). Multiplying by the charged
indexed-operation overhead and adding setup gives the displayed matched-route
cost. \Cref{lem:interface-closure-principle} then places that route in the
interface-closed classical baseline. The final statement is the corresponding
competitiveness condition.
\end{proof}

Since \(A_{\mathrm{idx}}=\operatorname{polylog}(N)\), the competitiveness
condition is fixed by \(M^{\mathrm{imp}}_{\mathrm{Betti}}\) and therefore by the
estimator's dependence on the spectral gap \(\gamma\). For constant \(\gamma\)
the matched route runs in \(\operatorname{polylog}(N)\); for
\(\gamma=1/\operatorname{poly}(n)\) the imported estimator scales as
\(n^{\Theta(1/\sqrt{\gamma})}\) and the matched route inherits that growth in
\(N\).

\subsubsection{Membership-density obstruction}
\label{subsubsec:membership-density-obstruction}

\begin{proposition}[Clique-density obstruction for membership-based simplex preparations]
\label{prop:clique-density-obstruction}
Consider a declared TDA implementation whose available simplex-preparation
mechanism starts from the uniform superposition over all \((k+1)\)-subsets of
\([n]\) and uses a reversible clique-membership predicate to mark the valid
\(k\)-simplices. Let
\[
R_k(G)=\frac{\binom{n}{k+1}}{|S_k|}
\]
be the inverse density of \(k\)-simplices among all \((k+1)\)-subsets. Then the
direct matched classical sampling route induced by the declared membership
predicate is rejection sampling with expected \(R_k(G)\) membership tests per
accepted \(k\)-simplex. Therefore this declared package certifies a
polylogarithmic matched sampler only in regimes where \(R_k(G)\) is
polylogarithmic in \(N=|S_k|\), unless additional generative structure is
declared and charged.
\end{proposition}

\begin{proof}
The declared predicate decides whether a candidate \((k+1)\)-subset is a
clique. A classical transducer using only that predicate and public randomness
can sample a uniformly random \((k+1)\)-subset, test it, and accept it exactly
when it is a clique. Conditional on acceptance, the output is uniform over
\(S_k\), because the proposal distribution is uniform over the ambient subset
space and every valid clique is accepted with the same probability. The
acceptance probability is
\[
\frac{|S_k|}{\binom{n}{k+1}},
\]
so the expected number of trials is
\[
R_k(G)=\frac{\binom{n}{k+1}}{|S_k|}.
\]
Thus the direct branch-evaluation route certified by this package has this
overhead. A polylogarithmic matched sampler follows from this certified route
when \(R_k(G)\) is polylogarithmic in \(N\). This statement concerns the declared
route; additional structure may support other classical algorithms.
\end{proof}

\subsubsection{Concrete adjudication of TDA access declarations}
\label{subsubsec:tda-proposal-adjudication}
\label{subsec:tda-access-declarations}

\Cref{thm:tda-interface-trichotomy} classifies implementation interfaces, not
complete quantum TDA algorithms. Different proposals may use different spectral
primitives, phase-estimation routines, or trace estimators; the audit asks
which information is declared to realize the simplex register and Laplacian
access. \Cref{tab:proposal-audit} summarizes the resulting adjudication.

\begin{table}[ht]
\centering
\small
\setlength{\tabcolsep}{2.5pt}
\renewcommand{\arraystretch}{1.08}
\begin{tabularx}{\textwidth}{@{}L{0.18\linewidth}L{0.22\linewidth}L{0.22\linewidth}L{0.16\linewidth}Y@{}}
\toprule
Proposal / family & Declared simplex access & Declared Laplacian/spectral access &
Audit class & Interface-closed consequence \\
\midrule
Lloyd--Garnerone--Zanardi~\cite{lloyd_quantum_2016} &
Uniform ambient \((k+1)\)-subsets with clique membership/projection &
Spectral routine for the combinatorial Laplacian &
Membership-based &
Matched rejection route with overhead
\(\binom{n}{k+1}/|S_k|\). \\
\midrule
Gyurik--Cade--Dunjko~\cite{gyurik_towards_2020} &
Clique membership structure in the simplex preparation &
Quantum TDA spectral primitives under declared access assumptions &
Membership-based / declared-access dependent &
Density overhead is charged unless a stronger generator is declared. \\
\midrule
Akhalwaya et al.~\cite{akhalwaya_comparing_2024} &
Classical sampling/Monte Carlo access when declared &
Stochastic trace or Betti-estimation module &
Estimator module &
Used as an imported classical route once matched access is certified. \\
\midrule
McArdle--Gilyen--Berta~\cite{mcardle_streamlined_2022} &
Structured simplex encoding; cheap unranking is an additional declared feature &
Streamlined spectral/TDA routines &
Candidate indexed structure &
Indexed certificate applies only when a reversible generator and local row
routines are declared and charged. \\
\midrule
Bounded-treewidth indexed family &
Exact owner-based unranking of \(S_k\) &
Local combinatorial Laplacian transition/value routines &
Reversible indexed &
Witnesses that the indexed class is nonempty; cheap indexing implies
\(N=\operatorname{poly}(n)\). \\
\bottomrule
\end{tabularx}
\caption{Concrete TDA access declarations placed in the interface audit. The
classification concerns the declared implementation interface and the
transcripts it licenses.}
\label{tab:proposal-audit}
\label{tab:tda-interface-audit}
\end{table}

\subsubsection{Scope of the indexed certificate: bounded-treewidth realization}
\label{subsubsec:bounded-treewidth-witness}

The reversible indexed class \Irev{} of \Cref{def:rev-indexed-tda} is an
abstract specification. We now exhibit a concrete, structured family that
realizes it with an \emph{exact} bijection (hence $\eta_{\mathrm{samp}}=0$), and
then show that the very cost condition making the induced classical access cheap
forces the simplex count to be polynomial in the vertex number. The construction
therefore provides a witness of realizability together with a sharp
boundary for the regime in which the certificate operates.

Fix $q=k+1$. Let $G$ be a graph on $n$ vertices of treewidth at most $w$, and let
$(T,\{\mathcal{B}_t\}_{t\in[M]})$ be a tree decomposition with $|\mathcal{B}_t|\le w+1$
for all $t$, rooted and indexed in a fixed DFS pre-order. Every clique of $G$ is
contained in some bag, and the bags containing a fixed clique $\sigma$ form a
connected subtree of $T$. Define the \emph{canonical owner}
\[
  \own(\sigma)=\min\{\,t : \sigma\subseteq\mathcal{B}_t\,\},
\]
the minimum taken in the pre-order. For each bag set
\[
  L_t=\{\,\sigma\subseteq\mathcal{B}_t : |\sigma|=q,\ \sigma\text{ is a clique},\ \own(\sigma)=t\,\},
\]
so that the $k$-simplices partition exactly as $S_k=\bigsqcup_{t\in[M]}L_t$. Write
$c_t=|L_t|$, $P_t=\sum_{s\le t}c_s$ (in pre-order), and $P_M=N=|S_k|$.

\begin{lemma}[Bounded-treewidth realization of \Irev]\label{lem:tw-realization}
Let $G$, $(T,\{\mathcal{B}_t\})$, $q=k+1$ be as above. There is a reversible
indexed implementation $I_Q(x)\in\Irev$ whose generator
$G_{\mathrm{simp}}\colon[N]\to S_k$ is the exact bijection
\[
  i\in\{0,\dots,N-1\}\ \longmapsto\ \text{the }j\text{-th element of }L_t,
  \qquad P_{t-1}\le i<P_t,\quad j=i-P_{t-1},
\]
where the unique owner bag $t$ is found by binary search on $(P_t)_t$. Its costs
satisfy
\[
  T_{\mathrm{count}}=O\!\Big(M\binom{w+1}{q}\operatorname{poly}(w,\log n)\Big),
  \qquad
  T_{\mathrm{gen}}=\widetilde{O}\!\Big(\log M\cdot\log N+\binom{w+1}{q}\,k\log n\Big).
\]
In particular, evaluating $G_{\mathrm{simp}}$ on $i\sim\Unif[N]$ yields the exactly
uniform distribution on $S_k$, so the induced matched access of
\Cref{thm:tda-indexed-collapse} holds with $\eta_{\mathrm{samp}}=0$.
\end{lemma}

\begin{proof}
The owner test is local under the pre-order: since the bags containing $\sigma$
form a connected subtree, $\own(\sigma)=t$ iff $\sigma\subseteq\mathcal{B}_t$ and
either $t$ is the root or $\sigma\not\subseteq\mathcal{B}_{\mathrm{parent}(t)}$.
Thus each $L_t$ is computed by enumerating the $\binom{w+1}{q}$ candidate
$q$-subsets of $\mathcal{B}_t$, testing each for cliqueness ($O(q^2)$ edge
lookups) and ownership ($O(qw)$ to test containment in the parent bag); summing
over bags and prefix-summing gives $T_{\mathrm{count}}$ and the array $(P_t)_t$.
Because owners partition $S_k$, each $k$-simplex is produced exactly once, so
$G_{\mathrm{simp}}$ is a bijection $[N]\to S_k$. On input $i$, a binary search over
$(P_t)_t$ ($O(\log M)$ comparisons on $O(\log N)$-bit integers) locates $t$, after
which the $j$-th element of $L_t$ is materialized as a sorted $q$-tuple in
$O(\binom{w+1}{q}\,k\log n)$ time. All operations are reversible classical
circuits on computational-basis inputs, so the implementation lies in \Irev{} by
\Cref{def:rev-indexed-tda}, and the exact-bijection clause of
\Cref{thm:tda-indexed-collapse} gives $\eta_{\mathrm{samp}}=0$. The local
transition/value routines $O_{nb},O_{val}$ for $\Delta_k$ are the standard
combinatorial up/down adjacency on $S_k$, which are likewise reversible classical
and local with $d_\Delta=O(\binom{w+1}{q}+kn)$.
\end{proof}

The family is structurally rich: it includes outerplanar graphs, cacti, and
series--parallel graphs (all treewidth $\le 2$), and every fixed-$w$ sparse family
with cycles. It is neither the full simplex nor a product structure, and---unlike
the membership-only preparation of \Cref{prop:clique-density-obstruction}---it
carries an exact index for $S_k$, so the inverse clique-density factor
$\binom{n}{k+1}/|S_k|$ never appears.

\begin{proposition}[Indexability--size incompatibility]\label{prop:tw-incompat}
Suppose the realization of \Cref{lem:tw-realization} has generation cost
$T_{\mathrm{gen}}=\polylog(N)$. Then necessarily $\binom{w+1}{k+1}=\polylog(N)$,
and since any graph admits a tree decomposition with $M=O(n)$ bags, one has
\[
  N=|S_k|\ \le\ M\binom{w+1}{k+1}\ =\ \widetilde{O}(n)\ =\ \operatorname{poly}(n).
\]
Consequently the boundary operators $\partial_k,\partial_{k+1}$ have
$\operatorname{poly}(n)$ dimensions, and the \emph{exact} Betti number
\[
  \beta_k=\bigl(|S_k|-\rank\partial_k\bigr)-\rank\partial_{k+1}
\]
is computable classically in $\operatorname{poly}(n)$ time by exact rank
computation over $\mathbb{Q}$ (e.g.\ fraction-free elimination), using exact
linear algebra over the boundary matrices.
\end{proposition}

\begin{proof}
The term $\binom{w+1}{k+1}\,k\log n$ in $T_{\mathrm{gen}}$ is $\polylog(N)$ only if
$\binom{w+1}{k+1}=\polylog(N)$. A tree decomposition may be assumed to have at most
$n$ non-redundant bags (delete any bag contained in a neighbour and contract,
repeat), so $M=O(n)$. Hence $N\le M\binom{w+1}{k+1}=O(n)\cdot\polylog(N)$. Taking
logarithms, $\log N\le \log O(n)+O(\log\log N)$, so $\log N=O(\log n)$ and
$N=\operatorname{poly}(n)$, indeed $N=\widetilde{O}(n)$. With $\operatorname{poly}(n)$-dimensional
rational boundary matrices, both ranks are computed exactly in $\operatorname{poly}(n)$
time, and $\beta_k$ follows from the Euler/rank identity for the $k$-th homology.
\end{proof}

\paragraph{Diagnostic role.}
\Cref{prop:tw-incompat} makes the scope of the reversible indexed
regime explicit. Whenever the induced classical access is genuinely cheap
($T_{\mathrm{gen}}=\polylog N$ via the above unranking), the instance is already
$\operatorname{poly}(n)$-sized and its exact Betti number is classically tractable by
linear algebra alone, so the comparison is already polynomial in \(n\). The
bounded-treewidth hypothesis is essential to this implication: the clique
complex of \(K_n\) admits \(\operatorname{polylog}(N)\) combinatorial unranking
with \(N=\binom{n}{k+1}=2^{\Theta(n)}\), so cheap indexing coexists with
exponential \(N\) outside bounded treewidth. Indexed families with
\(N=2^{\Omega(n)}\) and nontrivial \(\beta_k\), where the certificate carries
nontrivial end-to-end content, remain to be characterized. The regime that
motivates quantum TDA, namely \(N=\binom{n}{k+1}=2^{\Theta(n)}\) with
\(k=\Theta(n)\), forces \(\binom{w+1}{k+1}\)---and hence
\(T_{\mathrm{gen}}\)---to be exponential in the bounded-treewidth realization.
\Cref{thm:tda-indexed-collapse} certifies exactly the implementations whose
simplex preparation is a reversible classical generator: for those,
single-branch evaluation supplies the matched route. Preparations that obtain
the simplex register through coherent membership projection or amplitude
amplification fall under the membership-based class, whose directly certified
matched route is rejection sampling with the inverse clique-density overhead
\(\binom{n}{k+1}/|S_k|\), unless an additional generative structure is declared
and charged.

\subsection{Illustration: Low-Rank Access Separation}
\label{sec:lowrank}
\label{subsec:low-rank-application}

This illustration isolates the input-access mechanism behind the framework in a
purely low-rank matrix setting.  The example shows that the declared
access model determines the reducibility verdict for low-rank structure. For the same
structured low-rank family and the same recovery contracts, entrywise access
faces a localization barrier, while a mass-exposing \(\ell_2\) sample-and-query
interface exposes the relevant support directly.  The admissibility framework
then decides which of these routes belongs to the interface-closed baseline
from the declared implementation package.

\subsubsection{A block-sparse low-rank family}

Fix integers \(N,r,s\) with \(rs\leq N\), and let
\[
m = r s^2 .
\]
Let \(\mathcal C_{N,r,s}\) be the class of matrices \(A\in\mathbb R^{N\times N}\)
of the form
\[
A=\sum_{\ell=1}^r \sigma_\ell u_\ell v_\ell^\top ,
\]
where each \(u_\ell\) is supported on an unknown row set
\(R_\ell\subset [N]\) of size \(s\), each \(v_\ell\) is supported on an
unknown column set \(C_\ell\subset [N]\) of size \(s\), the row sets are
mutually disjoint, the column sets are mutually disjoint, all entries on
\(R_\ell\times C_\ell\) are nonzero, and
\[
\|u_\ell v_\ell^\top\|_F=1,
\qquad
0<c_\sigma\leq |\sigma_\ell|\leq C_\sigma<\infty .
\]
The supports are block-disjoint, so
\[
\operatorname{rank}(A)=r,
\qquad
|\operatorname{supp}(A)|=r s^2=m .
\]
The parameters \(c_\sigma,C_\sigma\) are absolute constants.  They ensure that
the Frobenius mass of the \(r\) blocks is balanced up to constant factors.

We compare two access models.  Entrywise access returns \(A_{ij}\) on a queried
coordinate \((i,j)\).  The \(\ell_2\) sample-and-query access model permits
sampling coordinates according to
\[
\Pr[(i,j)] = \frac{|A_{ij}|^2}{\|A\|_F^2},
\]
together with value queries to sampled or requested coordinates, with the
declared setup and sampling costs charged.

\subsubsection{Recovery contracts}

The first contract asks for a single nonzero witness:
\[
\mathrm{Red}^{\mathrm{wit}}_N(A,z)=1
\quad\Longleftrightarrow\quad
z=(i,j)\ \text{and}\ A_{ij}\neq 0 .
\]
The second contract asks for component representatives.  Let
\[
B_\ell = R_\ell\times C_\ell
\]
be the \(\ell\)-th nonzero block.  The component-representative contract is
satisfied by an output set \(Z\subset [N]\times [N]\) if
\[
Z\cap B_\ell\neq \varnothing
\qquad
\text{for every } \ell\in[r].
\]
The output may contain more than one representative per block.  The cost
measure below counts access queries and samples; arithmetic and output-writing
costs can only increase the total cost.

\begin{theorem}[Low-rank access non-invariance]
\label{thm:low-rank-access-noninvariance}
For the block-sparse low-rank class \(\mathcal C_{N,r,s}\), the witness and
component-representative contracts have different query complexities under
entrywise access and under \(\ell_2\) sample-and-query access.

For witness recovery,
\[
R^{\mathrm{wit}}_{\mathrm{entry}}
   (\mathcal C_{N,r,s},1/3)
=
\Omega\!\left(\frac{N^2}{r s^2}\right),
\qquad
R^{\mathrm{wit}}_{\ell_2\text{-}\mathrm{SQ}}
   (\mathcal C_{N,r,s},1/3)
=
O(1).
\]
For component-representative recovery,
\[
R^{\mathrm{comp}}_{\mathrm{entry}}
   (\mathcal C_{N,r,s},1/3)
=
\Omega\!\left(\frac{N^2}{s^2}\right),
\qquad
R^{\mathrm{comp}}_{\ell_2\text{-}\mathrm{SQ}}
   (\mathcal C_{N,r,s},\delta)
=
O\!\left(r\log\frac r\delta\right).
\]
The entrywise lower bounds hold for randomized algorithms with constant
success probability, and the \(\ell_2\)-sample-and-query upper bounds hold with
the stated success probabilities.
\end{theorem}

\begin{proof}
We use Yao's minimax principle for the entrywise lower bounds.

For witness recovery, draw the \(r\) row supports and \(r\) column supports
uniformly from valid disjoint configurations, and then draw block factors
satisfying the normalization and nonzero-support assumptions.  Fix a
deterministic adaptive entrywise algorithm making \(q\) queries.  Until a query
hits \(\operatorname{supp}(A)\), every answer is zero.  For every fixed queried
coordinate, the probability over the random block placement that it lies in the
support is
\[
\frac{m}{N^2}=\frac{r s^2}{N^2}.
\]
By the union bound,
\[
\Pr[\text{some query hits }\operatorname{supp}(A)]
\leq
\frac{q m}{N^2}.
\]
If no query hits the support, the transcript contains only zeros.  A final
guessed coordinate then lies in the still-hidden support with probability at
most \(m/(N^2-q)\), for \(q<N^2\).  Thus the success probability is bounded by
\[
\frac{q m}{N^2}+\frac{m}{N^2-q}.
\]
Constant success requires
\[
q=\Omega\!\left(\frac{N^2}{m}\right)
=
\Omega\!\left(\frac{N^2}{r s^2}\right).
\]
Yao's principle converts this distributional deterministic lower bound into a
randomized worst-case lower bound.

Under \(\ell_2\) sample-and-query access, all Frobenius mass lies on
\(\operatorname{supp}(A)\).  Therefore one sample from
\[
\Pr[(i,j)]=|A_{ij}|^2/\|A\|_F^2
\]
is a nonzero coordinate with probability one.  This proves the \(O(1)\)
witness upper bound.

For component-representative recovery under \(\ell_2\) sampling, the Frobenius
mass of block \(\ell\) is \(\sigma_\ell^2\), because
\(\|u_\ell v_\ell^\top\|_F=1\).  Since
\(c_\sigma\leq |\sigma_\ell|\leq C_\sigma\), a single \(\ell_2\) sample falls
in any fixed block with probability bounded above and below by constant
multiples of \(1/r\).  Standard coupon collection therefore gives
\[
O\!\left(r\log\frac r\delta\right)
\]
samples to touch all \(r\) blocks with probability at least \(1-\delta\).  The
algorithm outputs the sampled coordinates; with the stated probability this set
contains at least one representative from every block.

For the entrywise component lower bound, fix \(r-1\) blocks and hide the
remaining block uniformly among a family of disjoint candidate \(s\times s\)
row-column locations.  The number of such candidates is
\[
M=\Theta\!\left(\frac{N^2}{s^2}\right).
\]
A deterministic entrywise algorithm that has not queried inside the hidden
block receives the same transcript for all remaining candidates.  If it makes
\(q=o(M)\) queries, the probability of intersecting the hidden block is
\(o(1)\), and the probability of naming a representative from it remains
bounded away from a constant.  Constant success therefore requires
\[
q=\Omega(M)
=
\Omega\!\left(\frac{N^2}{s^2}\right).
\]
\end{proof}

The theorem gives a concrete access-relative separation.  The same rank-\(r\)
promise and the same recovery contract have a localization cost under
entrywise access and a direct support-exposure route under
\(\ell_2\) sample-and-query access. The operative difference is the transcript
distribution licensed by the declared access model.

\subsubsection{Interface-closure consequence}

Let \(A_{\ell_2\text{-}\mathrm{SQ}}\) denote the classical
\(\ell_2\) sample-and-query access model for the matrix \(A\).  Suppose an
applied quantum claim for a task over \(\mathcal C_{N,r,s}\) declares an
implementation package \(I_Q(A)\) containing a mass-exposing data structure,
norm tree, sampling table, amplitude data structure, or common representation
from which valid \(\ell_2\) sample-and-query transcripts are generated with
charged overhead
\[
\Gamma_{\ell_2\text{-}\mathrm{SQ}}(N,r,s,\varepsilon,\delta).
\]
Then
\[
A_{\ell_2\text{-}\mathrm{SQ}}\preceq_{\mathrm{int}} A_Q .
\]
By the interface-closure principle, the corresponding sample-and-query
algorithms from \Cref{thm:low-rank-access-noninvariance} enter the
interface-closed classical baseline.

For the witness contract, the matched route has cost
\[
T^{\mathrm{wit}}_M
=
O\!\left(\Gamma_{\ell_2\text{-}\mathrm{SQ}}\right),
\]
because one \(\ell_2\)-sample returns a nonzero witness.  Hence
\[
T^{\mathrm{E2E,wit}}_C
\leq
O\!\left(\Gamma_{\ell_2\text{-}\mathrm{SQ}}\right)
\]
under the interface-closed comparison.

For the component-representative contract, the matched route has cost
\[
T^{\mathrm{comp}}_M
=
O\!\left(
\Gamma_{\ell_2\text{-}\mathrm{SQ}}\,
r\log\frac r\delta
\right),
\]
up to the charged cost of writing the sampled representatives.  Therefore
\[
T^{\mathrm{E2E,comp}}_C
\leq
O\!\left(
\Gamma_{\ell_2\text{-}\mathrm{SQ}}\,
r\log\frac r\delta
\right)
\]
for the interface-closed baseline.

Under the declared entrywise access standard, the relevant certified reduction
costs are the localization bounds
\[
R^{\mathrm{wit}}_{\mathrm{entry}}
=
\Omega\!\left(\frac{N^2}{r s^2}\right),
\qquad
R^{\mathrm{comp}}_{\mathrm{entry}}
=
\Omega\!\left(\frac{N^2}{s^2}\right).
\]
Thus the framework assigns different end-to-end baseline verdicts to the same
low-rank family according to the declared implementation interface: a
mass-exposing package licenses the sample-and-query route, while an entrywise
declaration leaves the localization cost visible in the access-reduction
complexity.

\section{Discussion and Conclusion}
\label{sec:discussion}

The main conclusion is that applied quantum-advantage claims should be audited
at the level of declared implementation interfaces. A primitive speedup becomes
an end-to-end comparison through the input transcripts licensed by the
implementation package and the output contract required by the application. The
admissibility relation
\(A_M\preceq_{\mathrm{int}}A_Q\) formalizes this interface-level question. When
a matched classical access model is generated by the same declared package that
realizes the quantum access, the corresponding classical route belongs to the
interface-closed baseline with all setup, sampling, query, arithmetic,
precision, and amplification costs charged.

The main applied certificate developed here is the TDA interface audit for
additive normalized Betti estimation on clique complexes. The audit separates
three operational regimes. Abstract spectral declarations, such as an
unspecified block-encoding or state-preparation oracle for \(\Delta_k\), receive
matched-access certification from an accompanying implementation package,
transcript reduction, or shared representation. Membership-based simplex
preparations induce a rejection route whose cost exposes the inverse
clique-density factor
\[
    R_k(G)=\frac{\binom{n}{k+1}}{|S_k|}.
\]
Reversible indexed implementations certify matched simplex
sampling, local Laplacian row transcripts, and normalization information by
single-branch evaluation of the declared reversible routines. Combined with an
imported normalized-Betti estimator under this matched access, the resulting
route is an admissible member of the interface-closed classical comparison.
Exact Betti output, unnormalized Betti output, homology-basis output, dense
spectral output, and hidden-algebraic regimes require separate analysis.

The low-rank separation results support the same interface principle in a
simpler setting. The same low-rank object can present a localization barrier
under entrywise access and a cheap reduction under \(\ell_2\)
sample-and-query access. Thus structural parameters such as rank, sparsity, or
bond dimension determine reducibility together with the declared access model.
A meaningful
reducibility statement must specify the triple
\[
    (\mathcal C_N,A,\operatorname{Red}_N),
\]
together with the precision and failure probability demanded by the task. For
the low-rank QML and QLSA tasks covered by known quantum-inspired reductions,
the sufficient competitiveness rule applies under certified matched
sample-and-query access, compact output, the cited imported classical route,
and an interface-closed baseline. Finite-size utility, polynomial improvements,
and tasks governed by different reductions are separate questions.

The output-side results express the complementary constraint. A large Hilbert
space and a dense classical answer are distinct operational requirements. A
state prepared by a quantum primitive becomes a dense classical vector only
under a contract that asks for dense coordinates. Scalar expectations, samples,
local observables, dense vectors, and tomographic reconstructions therefore
carry different end-to-end costs. Dense output contracts can multiply the full
preparation--primitive--measurement cycle, while compact contracts preserve the
possibility of advantage. Dense-output QLSA/HHL serves here as a calibration
example for output contracts. A positive quantum advantage claim requires an
external classical lower bound under the declared access model.

Overall, the framework gives an operational audit for end-to-end inheritance.
On the input side, it identifies declared quantum implementations that certify
a matched classical route, expose a density or preprocessing overhead, or
require an additional interface certificate. On the output side, it separates
compact contracts from dense extraction tasks. Applied quantum advantage is
therefore evaluated by the primitive, the declared interface, the admissible
baseline family, and the operational output contract.

\section*{Acknowledgments}
This work was supported by the project Serverless4HPC (Grant
PID2023-152804OB-I00), funded by MICIU/AEI/10.13039/501100011033 and by the
European Regional Development Fund (ERDF/EU).

\bibliographystyle{plain}
\bibliography{references}

\appendix
\section{Proofs and Output Bounds}
\label{sec:methods}

\subsection{Proofs for access-relative reducibility}
\label{subsec:proof-low-rank-separation}

All lower bounds in this subsection are query lower bounds. Arithmetic and
output-writing costs can only increase the total cost. The randomized lower
bounds use Yao's minimax principle: it is enough to give a distribution over
instances on which every deterministic algorithm with too few queries has
success probability below the required threshold.

\begin{definition}[Block-sparse low-rank class]
\label{def:block-sparse-low-rank-class}
Let $\mathcal C_{N,r,s}$ be the class of matrices $A\in\mathbb R^{N\times N}$
of the form
\[
A=\sum_{\ell=1}^{r}\sigma_\ell u_\ell v_\ell^\top,
\]
where each $u_\ell$ is supported on an unknown row set
$R_\ell\subset[N]$ of size $s$, each $v_\ell$ is supported on an unknown column
set $C_\ell\subset[N]$ of size $s$, the row sets are mutually disjoint, the
column sets are mutually disjoint, all entries on $R_\ell\times C_\ell$ are
nonzero, and
\[
\|u_\ell v_\ell^\top\|_F=1,
\qquad
0<c_\sigma\le |\sigma_\ell|\le C_\sigma<\infty.
\]
Then $\rank(A)=r$ and $|\supp(A)|=rs^2$.
\end{definition}

\begin{theorem}[Witness recovery]
\label{thm:access-separation}
\label{thm:witness-recovery}
For the witness contract
\[
\operatorname{Red}^{\mathrm{wit}}_N(A,z)=1
\iff
z=(i,j)\text{ and }A_{ij}\ne 0,
\]
entrywise access and $\ell_2$ sample-and-query access satisfy
\[
R_{\mathrm{entry}}^{\mathrm{wit}}(\mathcal C_{N,r,s},1/3)
=
\Omega\!\left(\frac{N^2}{rs^2}\right),
\qquad
R_{\ell_2\text{-SQ}}^{\mathrm{wit}}(\mathcal C_{N,r,s},1/3)=O(1).
\]
The lower bound is worst-case over $\mathcal C_{N,r,s}$ and holds for
randomized algorithms with success probability at least $2/3$.
\end{theorem}

\begin{theorem}[Component recovery]
\label{thm:component-recovery}
Let the component-representative contract require an output set containing at
least one nonzero coordinate from each of the $r$ blocks. Under balanced
$\ell_2$ sample-and-query access,
\[
R_{\ell_2\text{-SQ}}^{\mathrm{comp}}(\mathcal C_{N,r,s},\delta)
=
O(r\log(r/\delta)).
\]
Under entrywise access, recovering a representative from every block has a
localization lower bound
\[
R_{\mathrm{entry}}^{\mathrm{comp}}(\mathcal C_{N,r,s},1/3)
=
\Omega\!\left(\frac{N^2}{s^2}\right).
\]
\end{theorem}

\begin{definition}[Block-local recoverability]
\label{def:block-local-recoverability}
The class $\mathcal C_{N,r,s}$ is block-locally recoverable under a value-query
model if, after receiving one nonzero coordinate from a block, an algorithm can
recover that block's row support, column support, and rank-one factors using
$L_{\mathrm{loc}}(s,\varepsilon,\delta)$ additional value queries and arithmetic
operations. When the surrounding data structure explicitly stores local block
membership, $L_{\mathrm{loc}}=O(s)$; under pair samples alone, additional
coupon-collection factors may be necessary.
\end{definition}

\begin{theorem}[Structured factor recovery]
\label{thm:structured-factor-recovery}
Consider the structured block-sparse factor recovery contract: output the $r$
blocks, their row and column supports, and rank-one factors
$\widehat u_\ell,\widehat v_\ell,\widehat\sigma_\ell$ such that the reconstructed
matrix $\widehat A$ satisfies
\[
\|A-\widehat A\|_F\le \eta\|A\|_F,
\qquad
\eta<\frac{c_\sigma}{2C_\sigma\sqrt r}.
\]
Then entrywise access has a localization lower bound
\[
R_{\mathrm{entry}}^{\mathrm{fac}}(\mathcal C_{N,r,s},\eta,1/3)
=
\Omega\!\left(\frac{N^2}{s^2}\right).
\]
If, in addition, block-local recoverability holds with
$L_{\mathrm{loc}}(s,\varepsilon,\delta)=O(s)$, then balanced
$\ell_2$ sample-and-query access satisfies
\[
R_{\ell_2\text{-SQ}}^{\mathrm{fac}}(\mathcal C_{N,r,s},\varepsilon,\delta)
=
O(r\log(r/\delta)+rs).
\]
The lower bound is for the stated structured block-recovery contract, with the
specified supports, representatives, and factor-output requirements.
\end{theorem}

\begin{proof}[Proof of \Cref{thm:witness-recovery}]
Let $m=rs^2=|\supp(A)|$. Draw the $r$ row sets and $r$ column sets uniformly
from valid disjoint configurations, and draw block factors satisfying the
normalization and nonzero-support assumptions. Fix a deterministic adaptive
entrywise algorithm making at most $q$ queries. Until a queried entry belongs to
$\supp(A)$, every answer is zero. For each fixed query, the probability over the
random block placement that the queried entry lies in the support is $m/N^2$.
The union bound gives
\[
\Pr[\text{some query hits }\supp(A)]\le \frac{qm}{N^2}.
\]
If no query hits the support, the transcript contains only zeros. Any final
guessed coordinate lies in the still-hidden support with probability at most
\[
\frac{m}{N^2-q}
\]
for $q<N^2$. Hence the success probability is at most
\[
\frac{qm}{N^2}+\frac{m}{N^2-q}.
\]
For success at least $2/3$, this implies
$q=\Omega(N^2/m)=\Omega(N^2/(rs^2))$. Yao's principle converts the
distributional deterministic lower bound into a randomized worst-case lower
bound.

For the $\ell_2$ sample-and-query upper bound, the sampling distribution is
\[
\Pr[(i,j)]=\frac{|A_{ij}|^2}{\|A\|_F^2}.
\]
All Frobenius mass lies on $\supp(A)$, so one sample is a nonzero witness with
probability one. This gives $O(1)$ sample complexity.
\end{proof}

\begin{proof}[Proof of \Cref{thm:component-recovery}]
The Frobenius mass of block $\ell$ is $\sigma_\ell^2$ because
$\|u_\ell v_\ell^\top\|_F=1$. Since
$c_\sigma\le |\sigma_\ell|\le C_\sigma$, a single $\ell_2$ sample falls in
block $\ell$ with probability between
\[
\frac{c_\sigma^2}{rC_\sigma^2}
\quad\text{and}\quad
\frac{C_\sigma^2}{rc_\sigma^2}.
\]
Thus block labels are coupons with probabilities bounded above and below by
constant multiples of $1/r$. Standard coupon collection gives
$O(r\log(r/\delta))$ samples to touch every block with probability at least
$1-\delta$. Recording one sampled coordinate from each block satisfies the
component-representative contract.

For the entrywise lower bound, fix $r-1$ blocks and hide the remaining block
uniformly among a family of disjoint candidate $s\times s$ row-column locations.
There are
\[
M=\Theta(N^2/s^2)
\]
such candidates. A deterministic entrywise algorithm that has not queried the
hidden block receives the same transcript for all remaining candidates. Under
the uniform distribution over candidates, a query intersects the hidden block
with probability at most $1/M$. If no query intersects it, the probability of
outputting a representative from the hidden block is at most $1/M$. Constant
success therefore requires $\Omega(M)=\Omega(N^2/s^2)$ queries.
\end{proof}

\begin{proof}[Proof of \Cref{thm:structured-factor-recovery}]
The lower bound uses the same hidden-final-block distribution as in the proof of
\Cref{thm:component-recovery}. A structured factor recovery satisfying
\[
\|A-\widehat A\|_F\le \eta\|A\|_F,
\qquad
\eta<\frac{c_\sigma}{2C_\sigma\sqrt r},
\]
cannot omit the hidden block. Indeed, omitting any block contributes Frobenius
error at least $c_\sigma$, while
\[
\|A\|_F=\left(\sum_{\ell=1}^r \sigma_\ell^2\right)^{1/2}
\le C_\sigma\sqrt r.
\]
Thus omission gives relative error at least
$c_\sigma/(C_\sigma\sqrt r)$, above the allowed threshold. An algorithm must
therefore localize the hidden block with constant probability, which requires
$\Omega(N^2/s^2)$ entrywise queries by the search argument above.
This is the structured block-sparse recovery lower bound of
\Cref{thm:structured-factor-recovery}.

For the $\ell_2$ sample-and-query upper bound, first use
\Cref{thm:component-recovery} to obtain one representative from each
block in $O(r\log(r/\delta))$ samples. Under block-local recoverability, each
representative allows recovery of the block support and rank-one factors using
$O(s)$ additional value queries and arithmetic operations. Repeating this for
all $r$ blocks gives
\[
O(r\log(r/\delta)+rs).
\]
\end{proof}

\subsection{Output lower bounds used in the main text}
\label{subsec:output-lower-bounds}
\label{subsec:output-lifting}

\begin{table}[t]
\centering
\small
\setlength{\tabcolsep}{3pt}
\begin{tabularx}{\textwidth}{@{}L{3.2cm}Y L{3.3cm} L{2.0cm}@{}}
\toprule
Contract & Conservative operational bound & Source / caveat & Status in $N$ \\
\midrule
Scalar expectation &
$O(\varepsilon^{-2}\log(1/\delta))$ by direct sampling; coherent amplitude
estimation changes query dependence under stronger access &
Sampling; amplitude estimation~\cite{brassard_quantum_2002,nayak_quantum_1999} &
compact \\
$k$ observables &
Conservative $O(k\varepsilon^{-2}\log(k/\delta))$; shadows improve dependence
for suitable observable families &
Classical shadows~\cite{huang_predicting_2020} &
compact if $k=\polylog N$ \\
Requested samples &
$M$ requested samples require $M$ repetitions &
Born sampling &
compact if $M=\polylog N$ \\
Full probability vector over $d$ outcomes &
$\Omega(d/\varepsilon^2)$ for standard distribution-estimation contracts &
Distribution estimation &
dense if $d=N^\alpha$ \\
Pure-state tomography in dimension $d$ &
$\Omega(d)$ to $\Omega(d/\varepsilon^2)$ depending on metric and contract &
Tomography~\cite{haah_sample-optimal_2017} &
dense if $d=N^\alpha$ \\
Rank-$r$ mixed-state tomography in dimension $d$ &
$\Omega(rd)$, with $\widetilde\Omega(rd/\varepsilon^2)$ under common stronger
metrics &
Low-rank tomography~\cite{yuen_improved_2023} &
dense if $rd=N^\alpha$ \\
Dense vector in $\mathbb R^N$ &
$\Omega(N)$ to write all coordinates; stronger bounds require a specified
reconstruction metric &
Output size / reconstruction &
dense \\
\bottomrule
\end{tabularx}
\caption{Conservative output-contract bounds used in the main text. The table
records operational repetitions or extraction cost, rather than output bit
length alone.}
\label{tab:output-contract-bounds}
\end{table}

The table above uses conservative forms of standard lower bounds. Scalar
expectation estimation by independent sampling has
$\varepsilon^{-2}\log(1/\delta)$ dependence, while amplitude estimation can
improve coherent query dependence under stronger oracle assumptions. For
$k$ observables, the direct union-bound estimate is
$O(k\varepsilon^{-2}\log(k/\delta))$; classical-shadow bounds can improve the
dependence for suitable observable classes~\cite{huang_predicting_2020}. Full
distribution estimation over $d$ outcomes has sample complexity linear in $d$
up to metric-dependent precision factors. Tomography lower bounds depend on the
state class and metric, but pure-state and rank-$r$ mixed-state tomography are
at least linear in the number of free parameters in the regimes used here
\cite{haah_sample-optimal_2017,yuen_improved_2023}. Dense vector output has an
\(\Omega(N)\) write-out lower bound independently of tomography.

\section{Canonical Case Summary}
\label{app:case-battery}

Table~\ref{tab:canonical-battery} summarizes the canonical regimes used as
orientation in the paper. It is not an additional formal test; it records which
of the two mechanisms in the main text is doing the work.

\begin{table}[t]
\centering
\scriptsize
\renewcommand{\arraystretch}{0.9}
\setlength{\tabcolsep}{2pt}
\begin{tabularx}{\textwidth}{@{}L{2.7cm}L{3.0cm}Y L{2.6cm}@{}}
\toprule
Case & Contract regime & Mechanism in this paper & Evidence status \\
\midrule
Shor / discrete logarithm &
Classical group input / compact certificate &
Compact output and no matched efficient classical route preserve the usual
hidden-subgroup inheritance story. &
Fault-tolerant resources explicit for factoring~\cite{gidney_how_2021}. \\
Low-rank QML / QPCA / recommendation &
State preparation or sample-and-query / compact prediction, spectrum, or
sample &
Mass-exposing access licenses a matched classical route, so covered tasks have
no superpolynomial gap under matched access. &
Quantum-inspired reductions and robust dequantization
\cite{tang_quantum-inspired_2019,gilyen_quantum-inspired_2018,
chia_sampling-based_2022,bakshi_improved_2024,le_gall_robust_2025}. \\
Normalized spectral trace, local operator &
\((1/N)\operatorname{tr} f(A)\), \(\sigma\)-smoothed scalar &
Local-probe estimators give a background module when the declared access
actually induces uniform index sampling and local row transcripts. &
Estimator module imported
\cite{cohen_steiner_approximating_2018,braverman_sublinear_2022}; not the
central access certificate of this paper. \\
Dense-output QLSA / HHL / Markowitz &
Dense matrix or KKT access / dense vector or portfolio output &
Dense loading and dense extraction block polylogarithmic end-to-end inheritance
when the task asks for a classical dense object. &
Classical linear-algebra and convex baselines; covariance structure and
benchmarks in finance. \\
TDA Betti, indexed regime &
Reversible indexed clique complex / normalized Betti (additive) &
Single-branch evaluation of reversible simplex-generation and local Laplacian
routines induces the matched sampler and local row transcripts. &
Explicit reversible indexed certificate in
\Cref{thm:tda-indexed-collapse}; TDA interface audit in
\Cref{thm:tda-interface-audit}; conditional competitiveness criterion in
\Cref{cor:tda-normalized-betti-bound}; clique-density boundary in
\Cref{prop:clique-density-obstruction}. \\
TDA Betti, residual non-certified regime &
Non-indexable simplex set / normalized Betti (additive) &
The indexed certificate is unavailable. A positive advantage claim in this
regime needs an external classical lower bound. &
Motivated by \(\mathrm{QMA}_1\)-hard clique homology
\cite{crichigno_clique_2022}. \\
Sampling / supremacy experiments &
Sampling distribution or cross-entropy style evidence &
The TDA interface audit does not apply, and no matched classical sampler
is certified by an abstract sampling task alone. &
Matching fails for this certificate; advantage may survive subject to the
usual sampling-complexity assumptions. \\
Hamiltonian simulation, local observables &
Compact local observables after simulation &
The dense-output obstruction need not apply, but the TDA interface audit is not
the relevant access claim. &
Matching fails for this certificate; any surviving advantage requires the
separate Hamiltonian-simulation lower-bound assumptions. \\
\bottomrule
\end{tabularx}
\caption{Canonical case summary. The rows are reminders of how the two main
mechanisms apply; they are not a separate predicate-evaluation protocol.}
\label{tab:canonical-battery}
\end{table}

\section{Extended Canonical Case Analyses}
\label{app:extended-cases}

The non-TDA canonical cases are summarized in Table~\ref{tab:canonical-battery},
where their essential evidence and citations are carried in the notes. This
appendix keeps the TDA analysis because it is the non-low-rank application where
the admissibility criterion is used explicitly.

\setcounter{subsection}{3}
\subsection{TDA Betti estimation as the boundary case}

Section~\ref{subsec:tda-audit} gives the main-body analysis of the boundary:
reversible indexed spectral TDA and persistent Betti estimation
\cite{lloyd_quantum_2016,mcardle_streamlined_2022} meet homological hardness
evidence
\cite{gyurik_towards_2020,schmidhuber_complexity_2022,crichigno_clique_2022}
and classical Monte Carlo or dequantized estimators
\cite{apers_simple_2022,berry_analyzing_2022,akhalwaya_comparing_2024}. The
contract split is the one proved there: reversible indexed simplex-generation
and local Laplacian routines induce a matched classical route for normalized
additive Betti output by single-branch evaluation, while membership-only or
hidden-algebraic regimes remain outside that particular certificate unless their
overheads are declared and charged. Any positive advantage claim still requires
independent lower-bound evidence.

\section{Reference Cost Models and Access Assumptions}
\label{app:cost-models}

\subsection{Access models}

\paragraph{Entrywise access.}
An entrywise oracle returns $A_{ij}$ for a requested pair $(i,j)$. It is weak for
detecting sparse hidden structure, as \Cref{thm:witness-recovery} illustrates.
When the input is dense and no additional structure is provided, constructing a
block-encoding or classical representation typically requires reading
$\Theta(N^2)$ entries.

\paragraph{Sparse-access oracles.}
Sparse matrix access gives row/column locations and values of nonzero entries. It can
make quantum simulation or QLSA input viable without implying low rank. Sparse full-rank
families are therefore distinct from low-rank sample-and-query families.

\paragraph{Block-encoding.}
A block-encoding represents a matrix in a larger unitary with normalization
$\alpha_A$. The end-to-end model charges both the construction/access cost $T_A$ and the
normalization factor $\alpha_A$ multiplying the primitive cost. Claims relying on a
block-encoding must state how it is obtained.
Matrix-vector product access is another declared interface, with lower bounds and
reconstruction guarantees that differ from entrywise and sample-and-query access
\cite{sun_querying_2019}.

\paragraph{QRAM and amplitude encoding.}
QRAM or amplitude encoding can load $N$ amplitudes into $O(\log N)$ qubits once the data
structure exists. The model charges the construction or access assumptions needed to make
that encoding available. A pre-existing amplitude oracle is not equivalent to entrywise
classical data unless the classical baseline receives a matched access model.

\paragraph{$\ell_2$ sample-and-query access.}
This access supports sampling indices proportional to squared row or entry norms and
querying selected values. It is the matched classical access behind several
quantum-inspired low-rank algorithms~\cite{tang_quantum-inspired_2019,
chia_sampling-based_2022}. Under this access, low-rank structure can be cheap to expose
classically.

\subsection{Output contracts}

Output contracts include scalar expectations, fixed collections of observables, samples,
compressed descriptions, dense vectors, pure-state reconstructions, and mixed-state
reconstructions. The contract model prices the output actually needed by the applied task.
Returning a quantum state for downstream quantum use is a different contract from
returning a classical vector. Sampling access is different from dense probability-vector
estimation. The output contract determines when a dense-looking state space
requires a dense classical output. Information-theoretic readout constraints begin with Holevo's
bound on classical information extractable from quantum states~\cite{holevo_bounds_1973}.
The conservative operational bounds used for classification are collected once
in Table~\ref{tab:output-contract-bounds}.

\subsection{Regime factors}

The regime multiplier $\rho_R$ summarizes physical overhead. In NISQ settings it may be
undefined or dominated by noise and trainability constraints; such cases are coded as
primitive or regime uncertainties. In early fault-tolerant and fault-tolerant settings,
$\rho_R$ may include error-correction cycle costs, magic-state factories, routing,
logical-to-physical qubit overhead, and precision-dependent repetition.
Threshold constructions give one standard source of such overheads
\cite{aharonov_fault-tolerant_1997}. The model does not prescribe a universal
$\rho_R$; it requires the claim to declare one or cite an end-to-end estimate.

\subsection{Scaling conventions}

The asymptotic scale must be declared with the claim: input length for
exponential-versus-polynomial comparisons, \(\log n\) for
polynomial-versus-polynomial comparisons, \(\log(1/\varepsilon)\) for precision
scaling, or the relevant search-space scale for black-box search. Changing the
scale changes the reported claim.

\end{document}